\documentclass[submission,copyright,creativecommons]{eptcs}
\providecommand{\event}{13th International Workshop on Foundations\\
 of Coordination Languages and Self-Adaptive Systems} 
           
\usepackage{amsthm}
\usepackage[english]{babel}
\usepackage{cite}
\usepackage{graphicx}
\usepackage{listings}
\usepackage{mathtools}
\usepackage{pgfplots}
\usepackage{tikz}
\usepackage{url}
\usepackage{amsmath}
\usepackage{amsfonts}
\usepackage{amssymb}
\usepackage{latexsym}
\usepackage{epsf}
\usepackage{gastex} 
\usepackage{localdef_B_D}
\usepackage{bach_D}
\usepackage{breakurl}  

\newcommand{\infrulemath}[2]{%
\begin{array}{c}
       {#1} \\
       \hline
       {#2}
\end{array}
}

\newtheorem{mydefinition}{\textbf{Definition}}
\newtheorem{myproposition}{\textbf{Proposition}}
\newtheorem{theo}{\textbf{Theorem}}

\title{On Distributed Density in Tuple-based Coordination Languages}
\author{Denis Darquennes \qquad \quad \qquad Jean-Marie Jacquet
\institute{Faculty of Computer Science \\
University of Namur 
\\ Namur, Belgium}
\email{denis.darquennes@unamur.be \qquad jean-marie.jacquet@unamur.be}
\and
Isabelle Linden
\institute{Department of Business Administration \\
Univeristy of Namur\\
Namur, Belgium}
\email{isabelle.linden@unamur.be}
}
\def\titlerunning{On Distributed Density in Tuple based Coordination Languages}
\def\authorrunning{D. Darquennes and all}
\begin{document}
\maketitle

\begin{abstract}

Inspired by the chemical metaphor, this paper proposes an extension of
Linda-like languages in the aim of modeling the coordination of
complex distributed systems. The new language manipulates finite sets
of tuples and distributes a density among them. This new concept adds
to the non-determinism inherent in the selection of matched tuples a
non-determinism to the tell, ask and get primitives on the
consideration of different tuples. Furthermore, thanks to de Boer and
Palamidessi's notion of modular embedding, we establish that this new
language strictly increases the expressiveness of the Dense Bach
language introduced earlier and, consequently, Linda-like languages.

\end{abstract}





\section{Introduction}
\label{jmj-introduction}

The technological evolutions over the last recent years confirm the
upward trends in a pervading of our everyday environments by new
mobile devices injecting or retrieving information from very dynamic
and dense networks. In order to garantee robustness and continuity of
the services they propose, the global structure must be tolerant
regarding any modification in topology, device technology or creation
of new services. These constraints will be fulfilled if full
self-organisation is incorporated as an inherent property of the
coordination models.  Coordination languages based on tuple spaces
offer an elegant response to such constraints. The Bach language - a
dialect of Linda developed at the University of Namur - is one of
them, and permits to model in an elegant way the interaction between
different components through the deposit and retrieval of tuples in a
shared space. As its basic form only allows the manipulation of one
tuple at a time and since the selection between several tuples
matching a required one is provided in an non-deterministic fashion, a
first extension was first proposed in \cite{Jacquet-Linden-Darquennes-Foclasa13}
in the aim of enriching traditional data-based coordination languages
by a notion of density attached to tuples, thereby yielding a new
coordination language, called Dense Bach.

To illustrate the use of the dense tuples, let us consider the context of service oriented computing. Let us imagine a situation
where a group of n researchers planing their presence to a conference want to book rooms in a hotel.
Their query could naturally consist in getting n rooms in the tuple space of free rooms
in any hotel, and to effectivelly book them only if all the n rooms are available, or if the number
of rooms combined with their number of beds meets the required n amount. In a more 
scientific context, we could consider the chemical reactions between three elements N (nitrogen), O (oxygen) and S (sulphur) present in a reactor.
Following the distribution - understood here in the chemical sense - of the global density between those three different reactants, some reactions will be more facilitated and others minder, and the concentration of the resulting products could alternate more in favour of one solution (like $NO_2$) or another (like $SO_2$).  

Considering those two previous examples, we then propose in this paper as a next natural step to consider a set
of tuples, among which the density is distributed. The new abstract
language resulting from that extension strictly increases the
expressiveness of Linda-like languages. It is still built upon the
four primitives tell, ask, get and nask, accessing a tuplespace, also
named subsequently store. However, it enhances it with a non-deterministic
behavior of the tell, ask and get primitives with regard to which
tuples from the set
 are considered.

Our purposes remain of a theoretical nature and, hence, for
simplicity purposes, we shall consider in this paper a simplified version where
tuples are taken in their simplest form of flat and unstructured
tokens. Nevertheless, the resulting simplification of the matching
process is orthogonal to our purposes and, consequently, our results
can be directly extended to more general tuples.

This paper fits into the continuity of previous work done by the
authors, among others of
\cite{Br-Ja-COORD97,BrJa-SCP00,BrJaLi-Foclasa02,LindenJBB06,Jacquet-Linden-Darquennes-Foclasa13}. As
a result, our approach follows the same lines of research, and employs
de boer and Palamidessi's modular embedding to test the expressiveness
of languages. The rest of this paper is consequently organized as
follows. Section 2 presents our extension of the Dense Bach language,
called Dense Bach with Distributed Density and, after the definition
of the distribution of density on a list of tokens, defines an
operational semantics. Section 3 provides a short presentation of
modular embedding and, on that basis, proceeds with an exhaustive
comparison of the relative expressive power of the languages Dense
Bach and Dense Bach with Distributed Density. Finally, section 4
compares our work with related work, draws our conclusions and
presents expectations for future work.


\section{Densed Tuple-based Coordination Languages}
\label{denseBachLang}

This section exposes in four points the different densed tuple-based coordination languages, firstly by presenting their
primitives, and secondly the different languages. The two last points present an operational semantic based on transition systems.   

\subsection{Primitives}

We start by defining the Bach and Dense Bach languages \cite{Jacquet-Linden-Darquennes-Foclasa13} from
which the language under study in this paper is an extension.

\subsubsection{Bach and Dense Bach}

The following definition formalizes how we attach a density to them.

\begin{mydefinition}
Let {\em Stoken} be a enumerable set, the elements of which are
subsequently called tokens and are typically represented by the
letters $t$ and $u$. Define the association of a token $t$ and a
positive integer $n \in \mathbb{N}$ as a \textit{dense
  token}. Such an association is typically denoted as $t(n)$. Define
then the set of dense tokens as the set {\em $SDtoken$}. Note that
since $Stoken$ and $\mathbb{N}$ are both enumerable, the set $SDtoken$ is
also enumerable.

Intuitively, a dense token $t(m)$ represents the simultaneous presence
of $m$ occurrences of $t$. As a result, $\lbrace t(m)\rbrace$
is subsequently used to represent the multiset 
$\lbrace t, \cdots, t\rbrace$
composed of these $m$ occurrences. 
Moreover, given two multisets of tokens
$\sigma$ and $\tau$, we shall use $\sigma \cup \tau$ to denote the multiset
union of elements of $\sigma$ and $\tau$. As a particular case,
by slightly abusing the syntax in writing $\{ t(m), t(n) \}$, we have
\( \lbrace t(m)\rbrace \cup \lbrace t(n)\rbrace = \lbrace t(m),t(n)\rbrace = \lbrace t(m + n)\rbrace
\).
Finally, we shall use $\sigma \uplus \lbrace t(m)\rbrace$ to denote,
on the one hand, the multiset union of $\sigma$ and $\lbrace
t(m)\rbrace$, and, on the other hand, the fact that $t$ does not
belong to $\sigma$. 
\end{mydefinition}

\begin{mydefinition} Define the set ${\cal T}$ of the token-based primitives
as the set of primitives $T$ generated by the following grammar:
\begin{eqnarray*}
T & ::= &
    tell(t) \ | \ ask(t)\ | \ get(t) \ | \ nask(t) 
\end{eqnarray*}
where $t$ represents a token.
\end{mydefinition}

\begin{mydefinition}
Define the set  
of dense token-based primitives ${\cal T}_d$
as the set of primitives $T_d$ generated by the following grammar:
\begin{eqnarray*}
T_d & ::= &
    tell(t(m)) \ | \ ask(t(m))\ | \ get(t(m)) \ | \ nask(t(m)) 
\end{eqnarray*}
where $t$ represents a token and $m$ a positive natural number.
\end{mydefinition}

The primitives of the Bach language are essentially the Linda ones
rephrased in a constraint-like setting. As a result, by calling
\textit{store} a multiset of tokens aiming at representing the current content
of the tuple space, the execution of the $tell(t)$ primitives amounts
to enrich the store by an occurrence of $t$. The $ask(t)$ and $get(t)$
primitives check whether $t$ is present on the store with the latter removing
one occurrence. Dually, $nask(t)$ tests whether $t$ is absent from the
store.

The primitives of the dense Bach language extend these primitives by
simultaneously handling multiple occurrences. Accordingly, $tell(t(m))$ atomically
puts $m$ occurrences of $t$ on the store and $ask(t(m))$ together with
$get(t(m))$ require the presence of at least $m$ occurrences of $t$
with the latter removing $m$ of them. Moreover, $nask(t(m))$ verifies
that there are less than $m$ occurrences of $t$.

\begin{figure}[t]
\begin{center}
\begin{minipage}{9cm}
\begin{scriptsize}
\[ \begin{array}{r@{\hspace*{0.25cm}}c}
     {\bf(T)}
& 
    $$\langle \; tell(t) \; \vert \; \sigma \;\rangle \longrightarrow \langle \; E \;\vert \; \sigma \cup \lbrace t\rbrace \;\rangle$$\\ 
\\ \\ 
     {\bf(A)}        
& 
     $$\langle \; ask(t) \;\vert\; \sigma \cup \lbrace t \rbrace \;\rangle \longrightarrow \langle\; E \;\vert\; \sigma \cup \lbrace t\rbrace \;\rangle$$\\ 
\\ \\ 
     {\bf(G)}        
& 
      $$\langle\; get(t)\;\vert\; \sigma \cup \lbrace t \rbrace \;\rangle \longrightarrow \langle\; E \;\vert\; \sigma \;\rangle$$\\
\\ \\ 
      {\bf(N)}        
& 
      $$\dfrac{t \not\in \sigma}{\langle\; nask(t)\;\vert\; \sigma \;\rangle \longrightarrow \langle\; E \;\vert\; \sigma \;\rangle}$$\\
\end{array}
\] 
\end{scriptsize}
\end{minipage}
\end{center}

\caption{Transition rules for token-based primitives (Bach)}
\label{jmj-fig-bach-based-primitives}
\end{figure}

\begin{figure}[t]
\begin{center}
\begin{minipage}{9cm}
\begin{scriptsize}
\[ \begin{array}{r@{\hspace*{0.25cm}}c}
     {\bf(T_d)}
& 
    $$\dfrac{m \in \mathbb{N}_{0}}{\langle \; tell(t(m)) \; \vert \; \sigma \;\rangle \longrightarrow \langle \; E \;\vert \; \sigma \cup \lbrace t(m)\rbrace \;\rangle}$$\\ 
\\ \\ 
     {\bf(A_d)}        
& 
     $$\dfrac{m \in \mathbb{N}_{0}}{\langle \; ask(t(m)) \;\vert\; \sigma \cup \lbrace t(m) \rbrace \;\rangle \longrightarrow \langle\; E \;\vert\; \sigma \cup \lbrace t(m)\rbrace \;\rangle}$$\\ 
\\ \\ 
     {\bf(G_d)}        
& 
      $$\dfrac{m \in \mathbb{N}_{0}}{\langle\; get(t(m))\;\vert\; \sigma \cup \lbrace t(m) \rbrace \;\rangle \longrightarrow \langle\; E \;\vert\; \sigma \;\rangle}$$\\
\\ \\ 
      {\bf(N_d)}        
& 
      $$\dfrac{n < m}{\langle\; nask(t(m))\;\vert\; \sigma \uplus \lbrace t(n) \rbrace \;\rangle \longrightarrow \langle\; E \;\vert\; \sigma \uplus \lbrace t(n) \rbrace \;\rangle}$$\\
\end{array}
\] 
\end{scriptsize}
\end{minipage}
\end{center}

\caption{Transition rules for dense token-based primitives (Dense Bach)}
\label{jmj-fig-token-based-primitives}
\end{figure}

These executions can be formalized by the transition steps of
figures~\ref{jmj-fig-bach-based-primitives} and
\ref{jmj-fig-token-based-primitives}, where configurations are pairs
of instructions, for the moment reduced to simple primitives, coupled
to the contents of a store. Note that $E$ is used to denote a
terminated computation. As can be seen by the above description, the
primitives of Bach are those of Dense Bach with a density of
1. Consequently, our explanation starts by the more general rules of
figure~\ref{jmj-fig-token-based-primitives}.  Rule $(T_d)$ states that
for any store $\sigma$ and any token $t$ with density $m$, the effect
of the tell primitive is to enrich the current set of tokens by $m$
occurrences of token $t$. Note that $\cup$ denotes multi-set
union. Rules $(A_d)$ and $(G_d)$ specify the effect of ask and get
primitives, both requiring the presence of at least $m$ occurrences of
$t$, but the latter also consuming them. Rule $(N_d)$ defines the nask
primitive, which tests for the absence of $m$ occurrences of $t$. Note
that there might be some provided there are less than $m$.  It is also
worth observing that thanks to the notation $\sigma \uplus \lbrace
t(n)\rbrace$ one is sure that $t$ does not occur in $\sigma$ and
consequently that there are exactly $n$ occurrences of $t$. This does
not apply for rules $(A_d)$ and $(G_d)$ for which it is sufficient to
assume the presence of at least $m$ occurrences, allowing $\sigma$ to
contain others.

Figure~\ref{jmj-fig-bach-based-primitives} specifies the transition
rules for the primitives of the Bach language. As expected, they
amount to the rules of Figure~\ref{jmj-fig-token-based-primitives}
where the density $m$ is taken to be 1 and the union symbol is interpreted
on multi-sets.

\subsubsection{Dense Bach with distributed density}

A natural extension is to replace a token by a set
 of tokens and to
distribute the density requirements on tokens. For instance, the
primitive $ask([t,u,v](6))$ succeeds on a store containing one
occurrence of $t$, two of $u$ and three of $v$. Dually, the computation of
$tell([t,u,v](6))$ may result in adding two occurrences of $t$ on the
store, three of $u$ and one of $v$. The following definitions formalize
this intuition.

\begin{mydefinition}
Let $Snlt$ denote the set of non-empty sets 
 of tokens in which, for
simplicity purposes, each token differs from the others.  Such a set 
is typically denoted as $L = [t_1, \ldots, t_p]$ and is thus such that
\(t_i \neq t_j\) for $i \not= j$. Define a dense set
 of tokens
as a set
 of $Snlt$ associated with a positive integer. Such
a dense set
 is typically represented as $L(m)$, with $L$ the set 
  of
tokens and $m$ an integer.
\end{mydefinition}

The distribution of the density over a set 
 of tokens is formalized
through the following distribution function.

\begin{mydefinition}
Define the distribution of tokens from dense sets 
 of tokens to sets of tuples of dense tokens as follows:
\[
     {\cal{D}}([t_1,\cdots,t_p](m)) 
     =  
     \lbrace  (t_{1} (m_{1}), \cdots, t_{p}(m_p)) : m_1 + \cdots + m_p = m \rbrace 
\]
Note that, thanks to the definition of dense tokens, we assume above
that the $m_{i}$'s are positive integers.  For the sake of simplicity,
we shall call the set ${\cal{D}}([t_1,\cdots,t_p](m))$ the
distribution of $m$ over $[t_1,\cdots,t_p]$.
\end{mydefinition}

The distribution of an integer $m$ over a set 
 of tokens $L$ has the
potential to express the behavior of the extended primitives. Indeed,
telling a dense set 
 amounts to telling atomically the $t_i(m_i)$'s of
a tuple defined above. Asking or getting a dense set 
 requires to
check that a tuple of ${\cal{D}}([t_1,\cdots,t_p](m))$ is
present on the considered store. For the negative ask, the requirement
is that none of the tuple is present. For the ease of writing and to
make this latter concept clear, we introduce the following concept of
intersection.

\begin{mydefinition}
Let $m$ be a positive integer,  $L=[t_1,\cdots,t_p]$ be a set 
 of tokens and $\sigma$ a store. We define 
\( {\cal{D}}(L(m)) \sqcap \sigma \) as the following set of tuples of dense tokens :
\[ {\cal{D}}(L(m)) \sqcap \sigma
   = 
   \lbrace (t_{1}(m_{1}), \ldots, t_{p}(m_p)) \in {\cal{D}}(L(m)) : \forall i \{ t_i (m_i) \} \subseteq \sigma
   \rbrace
\]
\end{mydefinition}

From an implementation point of view, it is worth observing that one
may give a syntactical characterization of the emptyness of such an
intersection.

\begin{mydefinition}
Given a store $\sigma$ and a dense set 
 $L(m)$, with
$L=[t_1,\cdots,t_p]$, we denote by $Max(\sigma,L(m))$ the tuple
$(t_{1}(m_{1}), \ldots, t_{p}(m_p))$ where the $m_i$'s denote the
number of occurrences of $t_i$ in $\sigma$. Moreover, we denote by
$SMax(\sigma,L(m))$ the sum
\( m_1 + \cdots + m_p
\)
\end{mydefinition}

It is easy to establish the following proposition.

\begin{sloppypar}
\begin{theo}
For any dense set 
 of tokens $L(m)$ and any store $\sigma$, one has
\( {\cal{D}}(L(m)) \sqcap \sigma = \emptyset
\)
iff 
\( SMax(\sigma,L(m)) < m
\).
\end{theo}
\end{sloppypar}

\begin{proof}
Simple verification.
\end{proof}

We are now in a position to specify the language extension handling dense sets 
 of tokens.

\begin{mydefinition}
Define the set of dense sets 
 primitives ${\cal T}_{dbd}$
as the set of primitives $T_{dbd}$ generated by the following grammar:
\begin{eqnarray*}
T_{dbd} & ::= &
    tell(L(m)) \ | \ ask(L(m))\ | \ get(L(m)) \ | \ nask(L(m)) 
\end{eqnarray*}
where $L(m)$ represents a dense set 
 of tokens.
\end{mydefinition}

The transition steps for these primitives are defined in
figure~\ref{fig-bach-based-primitives}. As suggested above, rule
$(T_{dbd})$ specifies that telling a dense set 
 $L(m)$ of tokens
amounts to atomically add the multiple occurrences $t_i(m_i)$'s of the
tokens of a tuple of the distribution of $m$ over $L$. Note that the
selected tuple is chosen non-deterministically, which gives to a tell
primitive a non-deterministic behavior as opposed to the tell
primitives of Bach and Dense Bach. Rule $(A_{dbd})$ states that
asking for the dense set 
 $L(m)$ amounts to testing that a tuple of
the distribution of $m$ over $L$ is in the store, which is technically
stated through the non-emptyness of the intersection of the
distribution and the store. Rule $(G_{dbd})$ requires that the tokens
of the tuples are removed in the considered multiplicity. Finally,
rule $(N_{dbd})$ specifies that negatively asking $L(m)$ succeeds if
$m$ is strictly positive and no tuple of the distribution of $m$ over
$L$ is present on the current store.

\begin{figure}[t]
\begin{center}
\begin{minipage}{9cm}
\begin{scriptsize}
\[ \begin{array}{r@{\hspace*{0.25cm}}c}
     {\bf(T_{dbd})}
& 
    $$\dfrac{ (t_1(m_1), \cdots, t_p(m_p)) \in {\cal D}(L(m))
           }{\langle \; tell(L(m)) \; \vert \; \sigma \;\rangle \longrightarrow 
             \langle \; E \;\vert \; \sigma \cup \{ t_1(m_1), \cdots, t_p(m_p) \}  \;\rangle
            }
    $$\\ 
\\ \\ 
     {\bf(A_{dbd})}        
& 
     $$\dfrac{ {\cal{D}}(L(m)) \sqcap \sigma \not=\emptyset
            }{ \langle \; ask(L(m)) \;\vert\; \sigma \;\rangle \longrightarrow 
               \langle\; E \;\vert\; \sigma \;\rangle
             }
     $$\\ 
\\ \\ 
     {\bf(G_{dbd})}        
& 
      $$\dfrac{ (t_1(m_1), \cdots, t_p(m_p)) \in {\cal D}(L(m))
             }{ \langle\; get(L(m))\;\vert\; \sigma \cup \lbrace t_1(m_1), \cdots, t_p(m_p) \rbrace \;\rangle 
                \longrightarrow \langle\; E \;\vert\; \sigma \;\rangle
              }
     $$\\
\\ \\ 
      {\bf(N_{dbd})}        
& 
      $$\dfrac{ m>0 \mbox{ and } {\cal{D}}(L(m)) \sqcap \sigma = \emptyset
             }{ \langle\; nask(L(m))\;\vert\; \sigma \;\rangle \longrightarrow \langle\; E \;\vert\; \sigma \;\rangle
              }
      $$\\
\\ \\

\end{array}
\] 
\end{scriptsize}
\end{minipage}
\end{center}

\caption{Transition rules for set of token-based primitives (Dense Bach with distributed Density)}
\label{fig-bach-based-primitives}
\end{figure}

\subsection{Languages}

We are now in a position to define the languages we shall consider.
The statements of these languages, also called {\em agents}, are
defined from the tell, ask, get and nask primitives by possibly
combining them by the classical choice operator $+$, used among others
in CCS, parallel operator (denoted by the $\parac$ symbol) and the
sequential operator (denoted by the $\seqc$ symbol). The formal
definition is as follows. 

\begin{mydefinition}
\label{jmj-lg-def}
\label{jmj-def-lg}
Define the Bach language ${{\cal L}_{B}}$ as the set of agents $A$ 
generated by the following grammar:
\begin{eqnarray*}
A & ::= &
    T \ | \ A\seqc A \ | \ A\parac A \ | \ A \choice A 
\end{eqnarray*}
where $T$ represents a token-based primitive. 
Define the Dense Bach language ${{\cal L}_{DB}}$ similarly but by taking dense
token-based primitives $T_d$:
\begin{eqnarray*}
A_d & ::= &
    T_d \ | \ A_d\seqc A_d \ | \ A_d\parac A_d \ | \ A_d \choice A_d 
\end{eqnarray*}
Define the Dense Bach with distributed Density language ${{\cal L}_{DBD}}$ similarly but by taking lists of
token-based primitives $T_{dbd}$:
\begin{eqnarray*}
A_{dbd} & ::= &
    T_{dbd} \ | \ A_{dbd}\seqc A_{dbd} \ | \ A_{dbd}\parac A_{dbd} \ | \ A_{dbd} \choice A_{dbd} 
\end{eqnarray*}
Subsequently, we shall consider sublanguages formed similarly but by
considering only subsets of these primitives. In that case, if
$\mathcal{H}$ denotes such a subset, then we shall write the induced
sublanguages as $\llg{${\cal H}$}$, $\llgBD{${\cal H}$}$, and
$\llgBDD{${\cal H}$}$ respectively. Note that for the latter
sublanguages, the tell, ask, nask and get primitives are associated
with the basic pairs described above.
\end{mydefinition}

\subsection{Transition system}

To study the expressiveness of the languages, a semantics needs to be
defined. As suggested in the previous subsections, we shall use an
operational one, based on transition systems.  For each transition
system, the configuration consists of agents (summarizing the current
state of the agents running on the store) and a multi-set of tokens
(denoting the current state of the store). In order to express the
termination of the computation of an agent, we extend the set of
agents by adding a special terminating symbol $E$ that can be seen as
a completely computed agent. For uniformity purpose, we abuse the
language by qualifying $E$ as an agent. To meet the intuition, we
shall always rewrite agents of the form ($E ; A$), ($E \parac A$) and
($A \parac E$) as $A$. This is technically achieved by defining the
extended sets of agents as ${{\cal L}_{B}} \cup \{E\}$, ${{\cal
    L}_{DB}} \cup \{E\}$ or ${{\cal L}_{DBD}} \cup \{E\}$ and by
justifying the simplifications by imposing a bimonoid structure.

\begin{figure}[t]
\begin{center}
\begin{minipage}{9cm}
\begin{scriptsize}
\[ \begin{array}{r@{\hspace*{0.25cm}}c}
     {\bf(S) }
   & 
     \infrulemath{  \transm{ \conf{ A }{ \sigma } 
                      }{ \conf{ A' }{ \sigma' } 
                       } 
            }{ 
                \transm{ \conf{ \seqcc{A}{B} }{ \sigma } 
                      }{ \conf{ \seqcc{A'}{B} }{ \sigma' } 
                       } 
              } 
   \\ \\ 
     {\bf(P) }
   & 
     \infrulemath{  \transm{ \conf{ A }{ \sigma } 
                      }{ \conf{ A' }{ \sigma' } 
                       } 
            }{ 
                 \begin{array}{c} 
                     \transm{ \conf{ \paracc{A}{B} }{ \sigma } 
                           }{ \conf{ \paracc{A'}{B} }{ \sigma' } 
                            } 
                  \\ 
                     \transm{ \conf{ \paracc{B}{A} }{ \sigma } 
                           }{ \conf{ \paracc{B}{A'} }{ \sigma' } 
                            } 
                  \end{array} 
              } 
    \\ \\ 
     {\bf(C) }
   & 
     \infrulemath{  \transm{ \conf{ A }{ \sigma } 
                       }{ \conf{ A' }{ \sigma' } 
                        } 
            }{ 
                 \begin{array}{c} 
                     \transm{ \conf{ \choicec{A}{B} }{ \sigma } 
                           }{ \conf{ A' }{ \sigma' } 
                            } 
                  \\ 
                     \transm{ \conf{ \choicec{B}{A} }{ \sigma } 
                           }{ \conf{ A' }{ \sigma' } 
                            } 
                  \end{array} 
              } 
\end{array} \]
\end{scriptsize}
\end{minipage}
\end{center}

\caption{Transition rules for the operators
\label{fig-operators}}
\end{figure}

The rules for the primitives of the languages have been given in
Figures~\ref{jmj-fig-bach-based-primitives} to \ref{fig-bach-based-primitives}. 
Figure~\ref{fig-operators} details the
usual rules for sequential composition, parallel composition,
interpreted in an interleaving fashion, and CCS-like choice.

\subsection{Observables and operational semantics}

We are now in a position to define what we want to observe from the
computations. Following previous work by some of the authors (see eg
\cite{Br-Ja-COORD99,BrJa-SCP00,Jacquet-Linden-Coord04,
Jacquet-Linden-Foclasa04,Jacquet-Linden-Foclasa03}),
we shall actually take an operational semantics recording the final
state of the computations, this being understood as the final store
coupled to a mark indicating whether the considered computation is
successful or not. Such marks are respectively denoted as $\delta^{+}$
(for the successful computations) and $\delta^{-}$ (for failed
computations).

\begin{mydefinition} 
\mbox{}

\begin{enumerate}

\item

Define the set of stores Sstore as the set of
finite multisets with elements from Stoken.

\item Let $\delta^{+}$ and $\delta^{-}$ be two fresh symbols
  denoting respectively success and failure. Define the set of
  histories Shist as the cartesian product \( Sstore \times
  \{\delta^{+},\delta^{-}\} \).

\item For each language ${\cal L}_I$ of the languages ${\cal L}_{B}$, ${\cal L}_{DB}$, ${\cal L}_{DBD}$,
define the operational semantics 
$ {\cal O}_I : {\cal L}_I \rightarrow {\cal P}(Shist)$
 as the following function: for any agent $A \in {\cal L}$
\begin{eqnarray*}
  {\cal O}(A)  & = &  \{(\sigma, \delta^{+}): \langle A | \emptyset \rangle \rightarrow^{*} \langle E | \sigma \rangle \}   \\
        &   & \mbox{} \cup \{(\sigma, \delta^{-}): \langle A | \emptyset \rangle \rightarrow^{*} \langle B | \sigma \rangle \nrightarrow, B \not= E \} 
\end{eqnarray*}

\end{enumerate}
\end{mydefinition}


\section{Comparison of Dense Bach and Dense Bach with Distributed Density}
\label{compDBandDBD}

This section focusses on the comparison between the Dense Bach language and the newly introduced
Dense Bach with Distrbuted Density language.

\subsection{Modular embedding}

\begin{figure}[t]
\begin{center}
\begin{scriptsize}
\begin{picture}(25,25)(8,8)

\gasset{Nadjust=w,Nadjustdist=2,Nh=8,fillcolor=White,Nfill=y,Nframe=n,ilength=7}

\node(Lprim)(10,30){${\mathcal L}'$}
\node(L)(10,10){${\mathcal L}$}
\node(Oprim)(30,30){${\cal O}'_{s}$}
\node(O)(30,10){${\cal O}_{s}$}

\drawedge[ELside=r](Lprim,L){${\cal C}$}
\drawedge(Lprim,Oprim){${\cal S}'$}
\drawedge[ELside=r](O,Oprim){${\cal D}_c$}
\drawedge(L,O){${\cal S}$}

\end{picture}
\end{scriptsize}
\end{center}
\caption{Basic embedding.\label{fig-embedding}}
\end{figure}

A natural way to compare the expressive
power of two languages is to determine whether all programs written in one
language can be easily and equivalently translated into the other
language, where equivalent is intended in the sense of conserving the same 
observable behaviors.

According to this intuition, Shapiro introduced in \cite{Shapiro} a
first notion of embedding as follows. Consider two
languages \(\mathcal{L}\) and \(\mathcal{L'}\). Assume given the
semantics mappings (\textit{Observation criteria}) \(\mathcal{S} :
\mathcal{L} \rightarrow \mathcal{O}_s \) and \(\mathcal{S'} :
\mathcal{L'} \rightarrow \mathcal{O'}_s \), where \(\mathcal{O}_s\) and
\(\mathcal{O'}_s\) are on some suitable domains. Then \( \mathcal{L} \) can
\textit{embed} \( \mathcal{L'} \) if there exists a mapping \(
\mathcal{C} \) (coder) from the statements of \( \mathcal{L'} \) to
the statements of \( \mathcal{L} \), and a mapping \( \mathcal{D}_c \)
(decoder) from \( \mathcal{O}_s \) to \( \mathcal{O'}_s \), such that the
diagram of Figure~\ref{fig-embedding} commutes, namely such that for
every statement \(A \in \mathcal{L'} :
\mathcal{D}_c(\mathcal{S}(\mathcal{C}(A))) = \mathcal{S'}(A) \).

This basic notion of embedding turns out however to be too weak since,
for instance, the above equation is satisfied by any pair of
Turing-complete languages. de Boer and Palamidessi hence proposed in
\cite{DB-PA-CONCUR-94} to add three constraints on the coder \(
\mathcal{C} \) and on the decoder \( \mathcal{D}_c \) in order to obtain
a notion of \textit{modular} embedding usable for concurrent
languages:

\begin{enumerate}

\item \(\mathcal{D}_c\) should be defined in an element-wise way with respect to \(\mathcal{O}_s\),
namely for some appropriate mapping \(\mathcal{D}_{el} \)
\begin{labeleqn}{\mbox{($P_1$)}}
\forall X \in {\cal O}_s: \ 
{\cal D}_c(X) = \{{\cal D}_{el}(x) \mid x \in X \}
\end{labeleqn}


\item the coder \(\mathcal{C}\) should be defined in a compositional way with respect to 
the sequential, parallel and choice operators:
\begin{labeleqn}{\mbox{($P_2$)}}
\begin{array}{c}
{\cal C}(A \seqc B) = {\cal C}(A) \seqc {\cal C}(B) \\
{\cal C}(A \parac B) = {\cal C}(A) \parac {\cal C}(B) \\
{\cal C}(A \choice B) = {\cal C}(A) \choice {\cal C}(B) \\
\end{array}
\end{labeleqn}


\item the embedding should preserve the behavior of the original processes with respect to deadlock, 
failure and success (\textit{termination invariance}):\\
\begin{labeleqn}{\mbox{$(P_3$)}}
\forall X \in {\cal O}_s, \forall x \in X: \ 
tm'({\cal D}_{el}(x)) = tm (x)
\end{labeleqn}%
%
where \textit{tm} and \textit{tm'} extract the termination information from the observables 
of \(\mathcal{L}\) and \(\mathcal{L'}\), respectively.
\end{enumerate} 

An embedding is then called \textit{modular} if it satisfies
properties \(\mathit{P}_{1}\), \(\mathit{P}_{2}\), and
\(\mathit{P}_{3}\).  The existence of a modular embedding from
\(\mathcal{L'}\) into \(\mathcal{L}\) is subsequently denoted by
\(\mathcal{L'} \leq \mathcal{L}\). It is easy to prove that \( \leq \)
is a pre-order relation. Moreover if \(\mathcal{L'} \subseteq
\mathcal{L}\) then \(\mathcal{L'} \leq \mathcal{L} \) that is, any
language embeds all its sublanguages. This property descends
immediately from the definition of embedding, by setting \(
\mathcal{C} \) and \( \mathcal{D}_c \) equal to the identity function.

\subsection{Formal propositions and proofs}

Let us now turn to the formal proofs. As a first result, thanks to the
fact that any language contains its sublanguages, a number of modular
embeddings are directly established.  In subsequent proofs, this is
referred to by {\em language inclusion}.

\begin{myproposition} 
\label{dda-prop-1a}
\( \llgBDD{ \(\psi\)} \leq \llgBDD{ \(\chi\)} \), for any subsets of
$\psi, \chi$ of primitives such that $\psi \subseteq \chi$.
\end{myproposition}

A second observation is that Dense Bach primitives are deduced from
the primitives of Dense Bach with Distributed Density by taking dense
sets 
 with only one token, the density being the same.  As a result
Dense Bach sublanguages are embedded in the corresponding Dense Bach
with Distributed Density sublanguages.

\begin{myproposition} 
\label{dda-prop-2a}
\( \llgBD{  \(\chi\)}   \leq   \llgBDD{  \(\chi\)}  \), 
for any subset of $\chi$ of primitives. 
\end{myproposition}

\begin{proof}
Immediate by defining the coder as follows:
\[ \begin{array}{llll}
\begin{array}{rcl}
    \coder(tell(t(m)))  & = & tell([t](m)) 
\\
    \coder(ask(t(m))) & = & ask([t](m))
\end{array}
& \hspace*{2cm} & 
\begin{array}{rclr}
    \coder(get(t(m))) & = & get([t](m)) &
\\
    \coder(nask(t(m)))  & = & nask([t](m)) & \hspace*{2cm} 
\end{array}
\end{array}
\]
\end{proof}

\begin{myproposition} 
\label{dda-prop-3a} 
\(\llgBD{tell}\) and \(\llgBDD{tell}\) are equivalent.
\end{myproposition}

\begin{sloppypar}
\begin{proof}
Indeed, thanks to proposition~\ref{dda-prop-2a}, 
\(\llgBD{tell} \leq \llgBDD{tell}\). Furthermore, as to each distribution of tokens from dense set 
 of tokens \( \cal{D}\)\(([t_{1},\cdots, t_{p}](m))\) is associated 
a finite set of tuple of dense tokens \(\lbrace  (t_{1} (m_{1}), \cdots, t_{p}(m_p)) : m_1 + \cdots + m_p = m \rbrace  \),  by coding
any $tell(L(m))$ primitive as \(tell(t_{1} (m_{1}))\seqc \cdots \seqc tell(t_{p}(m_p))\) primitives, 
and by using the identity as decoder, one establishes that
\(\llgBDD{tell} \leq \llgBD{tell}\).
\mbox{ }
\hfill 
\end{proof}
\end{sloppypar}

As a result of the expressiveness hierarchy of \cite{Jacquet-Linden-Darquennes-Foclasa13},
it also comes that both languages
\(\llgBD{ask,tell}\) and \(\llgBD{nask,tell}\) are strictly more
expressive than \(\llgBDD{tell}\) since both have been established
strictly more expressive than \(\llgBD{tell}\).

Let us now compare \( \llgBD{ask,tell} \) with its distributed dense counterpart.

\begin{myproposition} 
\label{dda-prop-4a} 
\( \llgBD{ask,tell}  <  \llgBDD{ask,tell} \)
\end{myproposition}

\begin{sloppypar}
\begin{proof}
On the one hand, 
\( \llgBD{ask,tell}  \leq  \llgBDD{ask,tell} \), 
by proposition~\ref{dda-prop-2a}. On the other hand,
\(\llgBDD{ask,tell} \not\leq \llgBD{ask,tell}\) may be established by contradiction.
The proof proceeds by exploiting the inability of $\llgBD{ask,tell}$ to
atomically test the presence of two distinct tokens $a$ and $b$.
Assume thus the existence of a coder 
\( {\cal C}: \llgBDD{ask,tell} \rightarrow \llgBD{ask,tell}
\) and
consider \( AB = ask([a,b](2))\). Let us prove that its coder is empty, which is absurd since,
by definition~\ref{jmj-lg-def}, it should contain at least one primitive. 
To that end,
one may assume that $\coder(AB)$ is in normal form
\cite{Br-Ja-Express98-ii} and thus is written as \(
tell(\dense{t_{1}});A_{1} + \cdots + tell(\dense{t_{p}});A_{p} +
ask(\dense{u_{1}});B_{1} + \cdots + ask(\dense{u_{q}});B_{q} \),
where $\dense{t_i}$ and $\dense{u_j}$ denote the token $t_i$ and $u_j$
associated to a density. In this expression, we will establish that there is no alternative guarded by a \(tell(\dense{t_{i}})\) operation, and no
alternative guarded by a \( ask(\dense{u_{j}}) \) operation either, in which case \(\coder(AB)\) is empty.

Let us first establish by contradiction that there is no alternative guarded by a \(tell(\dense{t_{i}})\) operation.
Assume there is one, say guarded by \(tell(\dense{t_{i}})\). Then
\(D = \langle \coder(AB)\vert\emptyset\rangle\rightarrow\langle A_{i}\vert { \dense{t_{i}}}\rangle \)
is a valid computation prefix of \(\coder(AB) \). It should deadlocks afterwards since
\( \opsem(AB)={ (\emptyset ,\delta^{-}) }\). However \( D \) is also a valid computation prefix of
\(\coder(AB + tell([a](1))) \). Hence, \(\coder(AB + tell([a](1))) \) admits a failing computation which
contradicts the fact that \(\opsem(AB  + tell([a](1))) = {(\{ a\},\delta^{+})} \).

Secondly we establish that there is also no alternative guarded by an
\( ask(\dense{u_{j}}) \) operation. To that end, let us first consider
two auxiliary computations: as \( \opsem( tell([a](1)) ) = { (\{ a \},
  \delta^{+}) } \), any computation of \( \coder( tell([a](1)) ) \)
starting in the empty store succeeds.  Let \( \langle
\coder(tell([a](1)))\vert\emptyset\rangle \rightarrow \dots
\rightarrow \langle E\vert \{ a_{1},\dots,a_{m}\}\rangle \) be such a
computation.  Similarly, let \( \langle
\coder(tell([b](1)))\vert\emptyset\rangle\rightarrow \dots
\rightarrow\langle E\vert \{ b_{1},\dots ,b_{n}\}\rangle \) be one
computation of \( \coder(tell([b](1)))\).  The proof of the claim
proceeds in two steps.  First let us prove that none of the $u_{i}$'s
belong to \( \{ a_{1},\dots ,a_{m} \} \).  By contradiction, assume
that $u_i = a_k$ for some $k$ and that $d$ is the density associated
to $u_i$, namely, $\dense{u_i} = u_i(d)$. Let us observe that, since
it is in $\llgBD{ask,tell}$, the considered computation of \(
\coder(tell([a](1))) \) can be repeated sequentially, as many times as
needed. As a result, by using $A^d$ to denote the sequential
composition of $d$ instances of $A$, the sequence \( D' = \langle
\coder(tell([a](1))^{d};AB)\vert \emptyset\rangle \rightarrow \dots
\rightarrow \langle AB\vert \{ a_{1}^d,\dots ,a_{m}^d \} \rangle
\rightarrow \langle B_{j}\vert \{ a_{1}^d,\dots ,a_{m}^d \}\rangle \)
is a valid computation prefix of \(\coder(tell([a](1))^{d};AB)\),
which can only be continued by failing suffixes. However \(D' \)
induces the following computation prefix \( D''\) for \(
tell([a](1))^{d};(AB + ask([a](1))) \) which admits only successful
computations: \(D'' = \langle \coder(tell([a](1))^d;(AB +
ask([a](1))))\vert \emptyset\rangle\rightarrow \dots \rightarrow
\langle AB + ask([a](1))\vert \{a_{1}^d,\dots
,a_{m}^d\}\rangle\rightarrow\langle B_{j}\vert \{a_{1}^d, \dots
,a_{m}^d\}\rangle\).  The proof proceeds similarly in the case \(
u_{j} \in { b_{1}, \dots ,b_{n}}\) for some \( j \in { 1, \dots ,q }\)
by then considering \( tell([b](1))^{d};AB \) and \(
tell([b](1))^{d};(AB + ask([b](1)))\).

Finally, the fact that the \( u_{i}'s \) do not belong to \( \{ a_{1}, \dots ,a_{m} \} \cup \{ b_{1}, \dots ,b_{n}\} \) induces a contradiction. Indeed, if this is the case then
\( \langle\coder(tell([a](1));tell([b](1));AB)\vert\emptyset\rangle \rightarrow \dots \rightarrow \langle tell([b](1));AB)\vert \{ a_{1}, \dots ,a_{m} \} \rangle \rightarrow \dots \rightarrow \langle AB \vert \{ a_{1}, \dots ,a_{m},b_{1}, \dots ,b_{n} \} \rangle \not\rightarrow \) is a valid failing computation prefix of
$\coder(tell([a](1));tell([b](1));AB)$ whereas \( tell([a](1));tell([b](1));AB \) has only one successful computation. As a conclusion, $\coder(AB)$ is equivalent to an empty statement, which is absurd by definition~\ref{jmj-def-lg}.
\end{proof}
\end{sloppypar}

Symmetrically, \( \llgBD{nask,tell} \) is strictly less expressive than \( \llgBDD{nask,tell} \).

\begin{myproposition} 
\label{dda-prop-5a} 
\( \llgBD{nask,tell} < \llgBDD{nask,tell} \).
\end{myproposition}

\begin{sloppypar}
\begin{proof}
On the one hand, \( \llgBD{nask,tell} \leq \llgBDD{nask,tell} \) holds
by proposition~\ref{dda-prop-2a}.  On the other hand, \(
\llgBDD{nask,tell} \not\leq \llgBD{nask,tell} \) is proved by
contradiction, assuming the existence of a coder ${\cal C}$.  The
proof proceeds as in proposition \ref{dda-prop-4a} but this time by
exploiting the inability of $\llgBD{nask,tell}$ to atomically test the
absence of two distinct tokens $a$ and $b$. 
\end{proof}
\end{sloppypar}


\(\llgBDD{nask,tell}\) and \(\llgBD{ask,tell}\) are not comparable with
each other, as well as \(\llgBDD{ask,tell}\) with regards to
\(\llgBD{nask,tell}\).

\begin{myproposition} 
\label{dda-prop-6a}
\[ \begin{array}{lll}
\begin{array}{ll}
(i) & \llgBDD{nask,tell} \not\leq  \llgBD{ask,tell} \\
(ii) & \llgBD{ask,tell} \not\leq  \llgBDD{nask,tell} 
\end{array}
& \hspace{10mm} &
\begin{array}{ll}
(iii) & \llgBDD{ask,tell} \not\leq  \llgBD{nask,tell} \\
(iv) & \llgBD{nask,tell} \not\leq  \llgBDD{ask,tell} 
\end{array}
\end{array}
\]
\end{myproposition}

\begin{proof}
\textbf{(i)} Otherwise we have \(\llgBD{nask,tell} \leq \llgBD{ask,tell}\)
which has been proved impossible in \cite{Jacquet-Linden-Darquennes-Foclasa13}. 
\textbf{(ii)} By contradiction, consider \( A = tell(t(1)) \seqc ask(t(1))) \).
One has \(\opsem(A) = \{ (\{ t(1) \},\delta^{+}) \} \). Hence, by \textit{$P_{3}$}, 
\( \coder(A) \) succeeds whereas we shall establish that it has failing
computations. Indeed, since 
\( \opsem(ask(t(1))) = \{ (\emptyset,\delta^{-}) \} \), 
any computation of \( \coder(ask(t(1))) \) starting on
the empty store fails. As \( \coder(ask(t(1))) \) is composed of nask and
tell primitives, this can only occur by having a nask primitive
preceded by a tell primitive. As enriching the initial content of the
store leads to the same result, any computation starting on any
(arbitrary) store fails. As a consequence, even if \( \coder(tell(t(1))) \)
has a successful computation, this computation cannot be continued by
a successful computation of \( \coder(ask(t(1))) \). Consequently any
computation of \( \coder(tell(t(1));ask(t(1))) \) fails, which produces a 
contradiction. 
\textbf{(iii)} Otherwise we would have \(\llgBD{ask,tell}
\leq \llgBD{nask,tell}\) which has been proved impossible in
\cite{Jacquet-Linden-Darquennes-Foclasa13}. 
\textbf{(iv)} By contradiction, consider 
\( A = tell(t(1)) \seqc nask(t(1))) \). One has
\( \opsem(A) = \{ (\{ t \},\delta^{-}) \} \). 
By \textit{$P_{3}$}, \( \coder(A) \) fails, whereas we shall establish that
it has a successful computation. Indeed, since 
\( \opsem(tell(t(1))) = \{ (\{ t(1) \},\delta^{+}) \} \), 
any computation of \( \coder(tell(t(1))) \)
starting on the empty store is successful. Similarly, it follows from
\( \opsem(nask(t(1))) = \{ (\emptyset,\delta^{+}) \} \) that any
computation of \( \coder(nask(t(1))) \) starting on the empty store is
successful, and, consequently, is any computation starting from any
store, since \( \coder(nask(t(1))) \) is composed of ask and tell
primitives. Summing up, any (successful) computation of \(
\coder(tell(t)) \) starting on the empty store can be continued by a
(successful) computation of \( \coder(nask(t)) \), which leads to the
contradiction.\hfill 
\end{proof}

\(\llgBDD{nask,tell}\) and \(\llgBDD{ask,tell}\) are not comparable with each other, as well as \(\llgBDD{nask,tell}\) with regards to \(\llgBD{ask,nask,tell}\).

\begin{myproposition} 
\label{dda-prop-7a}
\[ \begin{array}{lll}
\begin{array}{ll}
(i) & \llgBDD{nask,tell} \not\leq  \llgBDD{ask,tell} \\
(ii) & \llgBDD{ask,tell} \not\leq  \llgBDD{nask,tell} 
\end{array}
& \hspace{5mm} &
\begin{array}{ll}
(iii) & \llgBD{ask,nask,tell} \not\leq  \llgBDD{nask,tell} \\
(iv) & \llgBDD{nask,tell}  \not\leq  \llgBD{ask,nask,tell} 
\end{array}
\end{array}
\]
\end{myproposition}

\begin{proof}
\textbf{(i)} Otherwise \(\llgBD{nask,tell} \leq  \llgBDD{ask,tell}\), which contradicts proposition 6(iv). 
\textbf{(ii)} and \textbf{(iii)} Otherwise \(\llgBD{ask,tell} \leq  \llgBDD{nask,tell}\), 
which contradicts proposition 6(ii). 
\textbf{(iv)} The proof proceeds as in proposition 5(ii). The presence of the ask primitive 
in \( \LmultfBD \) 
 does not modify the reasoning, as it does not destroy elements and so does
 not modify the state of the store $\sigma$. \hfill 
\end{proof}

Symmetrically, \(\llgBD{get,tell}\) and \(\llgBDD{ask,tell}\) are not comparable with each other, as
 \(\llgBD{get,tell}\) and \(\llgBDD{nask,tell}\) are not comparable with each other.

\begin{myproposition} 
\label{dda-prop-8a}
\mbox{}

\noindent
\[ \begin{array}{lll}
\begin{array}{ll}
(i) & \llgBD{get,tell} \not\leq  \llgBDD{ask,tell} \\
(ii) & \llgBDD{ask,tell} \not\leq  \llgBD{get,tell} 
\end{array}
& \hspace{5mm} &
\begin{array}{ll}
(iii) & \llgBD{get,tell} \not\leq  \llgBDD{nask,tell} \\
(iv) & \llgBDD{nask,tell}  \not\leq  \llgBD{get,tell} 
\end{array}
\end{array}
\]
\end{myproposition}

\begin{sloppypar}
\begin{proof}
\textbf{(i)} By contradiction, consider \( tell(t(1)) \seqc get(t(1)) \).
One has \( \opsem(tell(t(1)) \seqc get(t(1))) = \{ (\emptyset,\delta^{+}) \}
\). By \textit{$P_{2}$} and \textit{$P_{3}$}, any computation of \(
\opsem( \coder(tell(t(1))) \seqc \coder(get(t(1))) ) \) is thus successful. Since
\( \coder(get(t(1))) \) is composed of ask and tell primitives only and
since ask and tell primitives do not destroy elements, at least one
computation of \( \opsem( \coder(tell(t(1))) \seqc \coder(get(t(1))) \seqc
\coder(get(t(1)))) \) is successful. However, \( \opsem( tell(t(1)) \seqc
get(t(1)) \seqc get(t(1)) ) = \{ (\emptyset, \delta^{-}) \}\), which provides 
the contradiction.

\textbf{(ii)} 
The proof is established by contradiction. Intuitively, $\llgBD{get,tell}$ is unable to atomically test the
presence of $a$ and $b$. Let us thus consider \( AB = ask([a,b](2))\) and prove that its coder has a successful
computation. This leads to a contradiction since $AB$ has just one failing computation.
To that end, one may assume that $\coder(AB)$ is in normal form (see \cite{Br-Ja-Express98-ii}) and thus is written as \(
tell(\dense{t_{1}});A_{1} + \cdots + tell(\dense{t_{p}});A_{p} + get(\dense{u_{1}});B_{1} + \cdots + get(\dense{u_{q}});B_{q} \),
where $\dense{t_i}$ and $\dense{u_j}$ denote the token $t_i$ and $u_j$ associated to a density.

The proof proceeds by establishing that
(I) there is no alternative guarded by a \( tell(\dense{t_{i}}) \) operation, and
(II) there is no alternative guarded by a \( get(\dense{u_{j}}) \) operation.
In which case, $\coder(AB)$ is equivalent to an empty statement, which
is not possible in view of definition~\ref{jmj-def-lg}.

CASE I: there is no alternative guarded by a \( tell(\dense{t_{i}}) \) operation. Otherwise, \(D = \langle \coder(AB)\vert\emptyset\rangle\rightarrow\langle A_{i}\vert \{ \dense{t_{i}}\}\rangle \) would be a valid computation prefix of \(\coder(AB) \) which should deadlocks afterwards since \( \opsem(AB)=\{ (\emptyset ,\delta^{-}) \}\).
However \( D \) is also a valid computation prefix of \(\coder(AB + tell([a](1))) \). Hence,
\(\coder(AB + tell([a](1))) \) admits a failing computation which contradicts the fact that
\(\opsem(AB  + tell([a](1))) = {(\{ a\},\delta^{+})} \).
  
CASE II: there is no alternative guarded by a \( get(\dense{u_{j}}) \) operation. To that end, let us first consider two auxiliary computations:
as \( \opsem( tell([a](1)) ) = \{ (\{ a \}, \delta^{+}) \} \),
any computation of \( \coder(tell([a](1))) \) starting in the empty store succeeds. 
Let \( \langle \coder(tell([a](1)))\vert\emptyset\rangle
\rightarrow \dots \rightarrow \langle E\vert\{ a_{1},\dots,a_{m}\}\rangle \)
be such a computation. Similarly, let 
\( \langle (tell([b](1)))\vert\emptyset\rangle
\rightarrow \dots \rightarrow\langle E\vert\{ b_{1},\dots ,b_{n}\}\rangle \) 
be one computation of \( \coder(tell([b](1)))\).
As these two computations start by assuming no token on the store and since
$\llgBD{get,tell}$ does not contain negative tests, it is easy to verify that
they can be put sequentially so as to establish the following computations:
\begin{eqnarray*}
   \langle \coder( tell([a](1)); tell([b](1)) ) 
           \vert
           \emptyset 
   \rangle 
   & \rightarrow \dots \rightarrow &
   \langle \coder( tell([b](1)) ) 
           \vert
           \{ a_{1},\dots,a_{m} \}
   \rangle \\
   & \rightarrow \dots \rightarrow &
   \langle E
           \vert
           \{ a_{1},\dots,a_{m}\} \cup \{ b_{1},\dots ,b_{n}\}
   \rangle
\\
   \langle \coder( tell([b](1)); tell([a](1)) ) 
           \vert
           \emptyset 
   \rangle 
   & \rightarrow \dots \rightarrow &
   \langle \coder( tell([a](1)) ) 
           \vert
           \{ b_{1},\dots,b_{n} \}
   \rangle \\
   & \rightarrow \dots \rightarrow &
   \langle E
           \vert
           \{ a_{1},\dots,a_{m}\} \cup \{ b_{1},\dots ,b_{n}\}
   \rangle
\end{eqnarray*} 
As \( \coder(tell([a](1));tell([b](1));AB)\)
has a successful computation, one of the $get(\dense{u_i})$ succeeds, and, consequently, one has
\( \{ \dense{u_j} \} \subseteq \{ a_{1},\dots,a_{m}\} \cup \{ b_{1},\dots ,b_{n}\}
\)
for some $j$. Assume $\dense{u_j} = a_k$ for $k$ and let $d$ be the density 
associated to $u_j$, namely, $\dense{u_j} = a_k(d)$. Then
\[ D' = 
   \langle \coder( tell([a](1)); AB ) 
           \vert
           \emptyset 
   \rangle
   \rightarrow \dots \rightarrow 
   \langle \coder( AB ) 
           \vert
           \{ a_{1},\dots,a_{m} \}
   \rangle
   \rightarrow
   \langle B_j
           \vert
           \{ a_{1},\dots,a_{m}\} \setminus \{ \dense{u_j} \}
   \rangle
\]
is a valid computation prefix of \(\coder(tell([a](1));AB)\).
It can only be continued by failing suffixes
since $tell([a](1)); AB$ fails. However, this induces
the following computation prefix $D''$ for 
\( \coder(tell([a](1));(AB + ask([a](1)))) \) and thus a failing computation
whereas \( tell([a](1));(AB + ask([a](1))) \) only admits a successful
computation: 
\begin{eqnarray*}
D'' = \langle \coder(tell([a](1));(AB + ask([a](1))))
                       \vert \emptyset\rangle 
&  \rightarrow \dots \rightarrow &
        \langle AB + ask([a](1))\vert \{a_{1},\dots ,a_{m}\}\rangle \\
&  \rightarrow & 
        \langle B_{j}\vert \{a_{1}, \dots ,a_{m} \} \setminus \{ \dense{u_j} \} \rangle.
\end{eqnarray*}
The proof proceeds similarly in the case \( u_{j} = b_k \) for some $k$
by then considering 
\( tell([b](1));AB \) and \( tell([b](1));(AB + ask([b](1)))\).
%
%
%
\textbf{(iii)} Otherwise, \(\llgBD{nask,tell}\) would be embedded in \(\llgBD{get,tell}\),
which has been prooved impossible in proposition 12(iv) of \cite{Jacquet-Linden-Darquennes-Foclasa13}.
\textbf{(iv)} Otherwise, \(\llgBD{ask,tell}\) would be embedded in \(\llgBDD{nask,tell}\) which contradicts
proposition 6(ii).

\end{proof}
\end{sloppypar}

Let us now prove that \(\llgBD{get,tell}\) is not comparable with \(\llgBDD{ask,nask,tell}\).

\begin{myproposition} 
\label{dda-prop-9a}
\mbox{}

\noindent
\( \begin{array}[t]{lll}
\\
(i) & & \llgBD{get,tell} \not\leq  \llgBDD{ask,nask,tell} \\
(ii) & & \llgBDD{ask,nask,tell}  \not\leq  \llgBD{get,tell}
\end{array}
\)
\end{myproposition}
\begin{sloppypar}
\begin{proof}
\textbf{(i)} By contradiction, consider \( tell(t(1)) \seqc get(t(1)) \), 
for which \(\opsem((tell(t(1)) \seqc get(t(1))) = \{ (\emptyset, \delta^{+}) \}\). 
Hence, by \textit{$P_{2}$} and \textit{$P_{3}$}, any computation of
\(\coder(tell(t(1))) \seqc \coder(get(t(1)))\) is successful. Such a
computation is composed of a computation for \(\coder(tell(t(1)))\)
followed by a computation for \(\coder(get(t(1)))\). As the latter is
composed of ask, nask,
and tell primitives which do not destroy elements
on the store, the latter computation can be repeated step by step
which yields a successful computation for \(\coder(tell(t(1))) \seqc
(\coder(get(t(1))) \parac \coder(get(t(1))))\). However,
\(\opsem(tell(t(1))\seqc(get(t(1)) \parac get(t(1))) = \{ (\emptyset,
\delta^{-}) \} \), which produces the announced contradiction.
\textbf{(ii)} Otherwise, \(\llgBD{nask,tell}\) would be embedded in \(\llgBD{get,tell}\) which has
 been prooved impossible in proposition 12(iv) of \cite{Jacquet-Linden-Darquennes-Foclasa13}.
\end{proof}
\end{sloppypar}

Let us now include the get primitive in the Dense Bach with Distributed Density language. We
first prove that \(\llgBDD{get,tell}\) is embedded in \(\llgBDD{ask,get,tell}\), but is not embedded in \(\llgBDD{ask,tell}\).

\begin{myproposition} 
\label{dda-prop-10a} 
\(  \llgBDD{get,tell} \leq \llgBDD{ask,get,tell}\) and \(\llgBDD{get,tell} \not\leq \llgBDD{ask,tell}\)
\end{myproposition}
\begin{sloppypar}
\begin{proof}
\textbf{(i)} One has
\( \llgBDD{get,tell} \leq  \llgBDD{ask,get,tell}
\)
by language inclusion. 
\textbf{(ii)} By contradiction, consider $A$ = \(tell([t](1)) \seqc
get([t](1)) \).  One has \( \opsem(A) = \{ (\emptyset ,\delta^{+}) \} \). 
By \textit{$P_{2}$} and \textit{$P_{3}$}, any computation of 
\( \opsem( \coder(tell([t](1))) \seqc \coder(get([t](1))))\) is thus
successful. Such a computation is composed of a computation for 
\( \coder(tell([t](1))) \) followed by a computation for 
\( \coder(get([t](1))) \). As \(\coder(get([t](1))) \) is composed of ask and
tell primitives and since ask and tell primitives do not destroy
elements, this latter computation can be repeated, which yields
successful computations for 
\(\opsem( \coder(tell([t](1))) \seqc
\coder(get([t](1))) \seqc \coder(get([t](1)))) \). 
However, 
\( \opsem( tell([t](1)) \seqc get([t](1)) \seqc get([t](1))) 
= \{ (\emptyset ,\delta^{-}) \} \), which leads to the contradiction. 

\hfill 
\end{proof}
\end{sloppypar}

\begin{sloppypar}
Let us now establish that \(\llgBD{get,tell}\) is strictly less expressive than \(\llgBDD{get,tell}\).
\end{sloppypar}

\begin{myproposition} 
\label{dda-prop-11a} 
\( \llgBD{get,tell} <  \llgBDD{get,tell}
\)
\end{myproposition}

\begin{proof}
On the one hand, 
\( \llgBD{get,tell} \leq  \llgBDD{get,tell}
\) 
holds by proposition~\ref{dda-prop-2a}. On the other hand,
\( \llgBDD{get,tell} \not\leq  \llgBD{get,tell}
\)
may be proved exactly 
as in proposition 8(ii), where we replace any occurrence of ask([a,b](2)) by get([a,b](2)).
\hfill 
\end{proof}

\begin{sloppypar}
In order to use once more the reasoning of proposition 8(ii),
 we now prove that \(\llgBDD{ask,tell}\) is
 not comparable with \(\llgBD{nask,get,tell}\).
\end{sloppypar}

\begin{myproposition} 
\label{dda-prop-12a}
\mbox{}

\noindent
\( \begin{array}[t]{lll}
\\
(i) & & \llgBDD{ask,tell} \not\leq  \llgBD{nask,get,tell} \\
(ii) & & \llgBD{nask,get,tell}  \not\leq  \llgBDD{ask,tell}
\end{array}
\)
%
%
\end{myproposition}

\begin{sloppypar}
\begin{proof}
%
%
\textbf{(i)} The proof proceeds as in proposition 8(ii), 
by constructing a successful coded
computation for the same failing agent 
\( ask([a,b](2)) \) 
with the alternatives guarded by a nask primitive of the normal form of the coded version treated as 
the alternatives guarded by a tell primitive.
\textbf{(ii)} Otherwise \(\llgBD{nask,tell} \leq
\llgBDD{ask,tell}\) which contradicts proposition 6(iv). \hfill
%
\end{proof}
\end{sloppypar}

\begin{sloppypar}
We can now prove that \(\llgBDD{get,tell}\) is not comparable with
respectively \(\llgBD{nask,tell}\), \(\llgBDD{nask,tell} \),
\(\llgBD{nask,get,tell}, \llgBDD{ask,nask,tell}\) and
\(\llgBD{ask,nask,tell}\).
\end{sloppypar}

\begin{myproposition} 
\label{dda-prop-13a}
\[ \begin{array}{lll}
\begin{array}{rl}
(i) & \llgBDD{get,tell} \not\leq \llgBD{nask,tell}  \\
(ii) & \llgBD{nask,tell} \not\leq \llgBDD{get,tell}  \\
(iii) & \llgBDD{get,tell} \not\leq \llgBDD{nask,tell}  \\
(iv) & \llgBDD{nask,tell} \not\leq \llgBDD{get,tell} \\
(v) & \llgBDD{get,tell} \not\leq \llgBD{nask,get,tell}
\end{array}
& \hspace{2mm} &
\begin{array}{rl}
(vi) & \llgBD{nask,get,tell} \not\leq \llgBDD{get,tell} \\
(vii) & \llgBDD{get,tell} \not\leq \llgBDD{ask,nask,tell} \\
(viii) & \llgBDD{ask,nask,tell} \not\leq \llgBDD{get,tell} \\
(ix) & \llgBDD{get,tell} \not\leq \llgBD{ask,nask,tell} \\
(x) & \llgBD{ask,nask,tell} \not\leq \llgBDD{get,tell}
\end{array}
\end{array}
\]
\end{myproposition}

\begin{sloppypar}
\begin{proof}
\textbf{(i)} Indeed, otherwise we have $\llgBD{ask,tell} \leq  \llgBD{nask,tell}$ 
which has been proved impossible in \cite{Jacquet-Linden-Darquennes-Foclasa13}. 
\textbf{(ii)} By contradiction, consider 
\( A = tell(t(1))\seqc nask(t(1))\), for which
\( \opsem(A) = \{ (\{ t(1) \}, \delta^{-}) \} \). Then,
by \textit{$P_{2}$} and \textit{$P_{3}$}, any computation of 
\( \coder(tell(t(1))) \seqc \coder(nask(t(1))) \) must fail whereas 
we shall establish that 
\( \coder(tell(t(1))) \seqc \coder(nask(t(1))) \) has a successful
computation. Indeed, let us observe that 
\( \opsem(tell(t(1))) = \{ (\{ t(1) \},\delta^{+}) \} \) and 
\( \opsem(nask(t(1))) = \{ (\emptyset,\delta^{+}) \} \). 
For both cases, by \textit{$P_{3}$},
any computation of \(\coder(tell(t(1))) \) and \(\coder(nask(t(1)))\)
starting on the empty store is successful. Consequently, since
\(\coder(tell(t(1))) \) and \(\coder(nask(t(1)))\) are 
composed of get and tell primitives, so are all of their 
computations starting from any store. Therefore, any
(successful) computation of \(\coder(tell(t(1)))\) starting on the empty
store can be continued by a (successful) computation of
\(\coder(nask(t(1)))\), which leads to the contradiction. 
\textbf{(iii)} Otherwise we have 
$\llgBD{ask,tell} \leq \llgBDD{nask,tell}$
which contradicts proposition 6(ii). 
\textbf{(iv)} Otherwise we have 
$\llgBD{nask,tell} \leq \llgBDD{get,tell}$ 
which contradicts (ii) above. 
\textbf{(v)} Similar to proposition 11.
%
\textbf{(vi)} Otherwise 
$\llgBD{nask,tell} \leq \llgBDD{get,tell}$ which contradicts (ii) above. 
\textbf{(vii)} By contradiction. Let us first observe that
\( \opsem(tell([t](1)) \seqc get([t](1)) ) = \{(\emptyset,\delta^{+}) \}\). 
By \textit{$P_{2}$} and \textit{$P_{3}$}
any computation of \((\coder(tell([t](1))) \seqc \coder(get([t](1))))\)
starting in the empty store is thus successful. By repeating step by step
the computation of $\coder(get([t](1)))$, this leads to 
a successful computation of
\((\coder(tell([t](1))) \seqc (\coder(get([t](1))) \parac
\coder(get([t](1)))))\) starting in the empty store. However, 
\(\opsem(tell([t](1)) \seqc (get([t](1)) \parac
get([t](1)))) = \{ (\emptyset,\delta^{-}) \}\), which leads to the
contradiction. 
\textbf{(viii)} Otherwise 
$\llgBDD{nask,tell} \leq \llgBDD{get,tell}$ 
which contradicts proposition (iv) above. 
\textbf{(ix)} Otherwise 
$\llgBDD{get,tell} \leq \llgBDD{ask,nask,tell}$ 
which contradicts (vii) above. 
\textbf{(x)} Otherwise $\llgBD{nask,tell} \leq \llgBDD{get,tell}$ which
contradicts (ii) above. 
\hfill 
\end{proof}
\end{sloppypar}

Let us now establish that \(\llgBDD{nask,tell}\) and
\(\llgBD{ask,nask,tell}\) are strictly less expressive than
\(\llgBDD{ask,nask,tell} \).

\begin{myproposition} 
\label{dda-prop-14a}
\mbox{}

\noindent
\( \begin{array}{ll}
\\
(i) & \llgBDD{nask,tell} < \llgBDD{ask,nask,tell} \\
(ii) &  \llgBD{ask,nask,tell} < \llgBDD{ask,nask,tell}
\end{array}
\)
\end{myproposition}

\begin{sloppypar}
\begin{proof}
\textbf{(i)} By sublanguage inclusion, one has 
\( \llgBDD{nask,tell} \leq \llgBDD{ask,nask,tell} \).
Moreover, if we had
\( \llgBDD{ask,nask,tell} \leq \llgBDD{nask,tell} 
\),
then we would have
$\llgBDD{ask,tell} \leq \llgBDD{nask,tell}$,
which contradicts proposition 7(ii).

\textbf{(ii)}
Let us thus proceed by contradiction by assuming the existence of a coder ${\cal C}$
from $\llgBDD{ask,nask,tell}$ to $\llgBD{ask,nask,tell}$.
Let $n$ be the cumulative occurrences of tokens in $nask$ primitives of $\coder(tell([a](1)))$.

As $\coder(tell([a](1)))$ has only successful computations, let
$S_a$ be the store
resulting from one of them. 
Moreover, as a matter of notation, 
let the construction $A^{||q}$ denote the parallel composition of $q$ copies of $A$.
As $(tell([b](1)))^{||(n+2)} \seqc
tell([a](1)))$ succeeds as well, let $S'_b$ denote the store resulting
from one successful computation of its coding.  Consider finally 
\( ABs = ask([a,b](n+4)) 
\)
with the intuitive aim of requesting one $a$ with
$n+3$ copies of $b$. Consider $\coder(ABs)$ in its normal form:
\[ \begin{array}{l}
      tell(\dense{t_1}) \seqc A_1 
     +   \cdots
     +   tell(\dense{t_p}) \seqc A_p \\ \mbox{}
     +   ask(\dense{u_1}) \seqc B_1
     +   \cdots
     +   ask(\dense{u_q}) \seqc B_q \\ \mbox{}
     +   nask(\dense{v_1}) \seqc C_1
     +   \cdots
     +   nask(\dense{v_r}) \seqc C_r
\end{array} \]
%
As in proposition 5(ii), it is possible to establish that there are no alternatives guarded by $tell(\dense{t_i})$ and $nask(\dense{v_j})$ primitives.
Let us
 prove that \( \{ u_1, \cdots, u_q \} \cap (S_a
\cup S'_b) = \emptyset \).  This is 
done in two steps by
establishing that (1) \( \{ u_1, \cdots, u_q \} \cap S_a = \emptyset
\), and that (2) \( \{ u_1, \cdots, u_q \} \cap S'_b = \emptyset \).

First let us prove that \( \{ u_1, \cdots, u_q \} \cap S_a = \emptyset
\). Assume $u_i \in S_a$ and let $d$ be the density associated to
$u_i$, namely, $\dense{u_i} = u_i(d)$. Let us observe that each step
of the considered computation of \( \coder(tell([a](1))) \) can be
repeated in turn, in as many parallel occurences of it as needed, so
that
\begin{eqnarray*}
  \lefteqn{P = \langle \coder( tell([a](1))^{||d} \seqc ABs) \vert \emptyset\rangle} \\
     & \rightarrow \dots \rightarrow & 
           \langle ABs \vert \cup_{k=1}^d S_a \rangle \\
     & \rightarrow & 
           \langle B_{i} \vert (\cup_{k=1}^d S_a) \rangle 
\end{eqnarray*}
is a valid computation prefix of 
\( \coder( tell([a](1))^{||q} \seqc ABs)
\),
which can only be continued by failing suffixes.
However \(P \) induces the following computation prefix \( P'\) for
\( \coder( tell([a](1))^{||q} \seqc (ABs + tell([a](1))) )
\)
which admits only successful computations:
\begin{eqnarray*}
  \lefteqn{P' = \langle \coder( tell([a](1))^{||d} \seqc (ABs + tell([a](1))) ) \vert \emptyset\rangle} \\
     & \rightarrow \dots \rightarrow & 
           \langle \coder(ABs+tell([a](1))) \vert \cup_{k=1}^d S_a \rangle \\
     & \rightarrow & 
           \langle B_{i} \vert (\cup_{k=1}^d S_a) \rangle 
\end{eqnarray*}

\noindent leading to the contradiction.

Secondly, the proof that \( \{ u_1, \cdots, u_q \} \cap S'_b =
\emptyset \) is established similarly by considering $S'_b$ instead of
$S_a$ and $tell([b](1))$ instead of $tell([a](1))$.
\end{proof}
\end{sloppypar}

\(\llgBDD{ask,tell}\) is strictly less expressive than \(\llgBDD{ask,nask,tell}\).

\begin{myproposition} 
\label{dda-prop-15a}
\( \llgBDD{ask,tell} < \llgBDD{ask,nask,tell}
\)
\end{myproposition}

\begin{sloppypar}
\begin{proof}
On the one hand, 
\( \llgBDD{ask,tell} \leq  \llgBDD{ask,nask,tell}
\) 
results from language inclusion.
On the other hand, one has
\( \llgBDD{ask,nask,tell} \not\leq  \llgBDD{ask,tell}
\) 
since otherwise 
\(\llgBDD{nask,tell} \leq \llgBDD{ask,tell}\),
which contradicts proposition 7(i).\hfill 
\end{proof}
\end{sloppypar}

\begin{sloppypar}
Symmetrically to proposition 12(i) and 12(ii),
\(\llgBD{nask,get,tell}\) is not comparable with \(\llgBDD{nask,tell}\).
\end{sloppypar}

\begin{myproposition} 
\label{dda-prop-16a}
\mbox{}

\( \begin{array}[t]{lll} \\
(i) & & \llgBD{nask,get,tell} \not\leq  \llgBDD{nask,tell}\\
(ii) & & \llgBDD{nask,tell} \not\leq \llgBD{nask,get,tell}
\end{array}
\)
\end{myproposition}

\begin{sloppypar}
\begin{proof}
\textbf{(i)} Otherwise, \(\llgBD{ask,tell} \leq \llgBDD{nask,tell}\)
which contradicts proposition 6(ii).

\textbf{(ii)}
The proof proceeds by contradiction, similarly to the proofs of 
\( \llgBDD{nask,tell} \not\leq \llgBD{ask,nask,tell}\) of proposition \ref{dda-prop-7a}(iv),
which itself extends that of
\( \llgBDD{nask,tell} \not\leq \llgBD{nask,tell}
\) of proposition \ref{dda-prop-5a}(ii).

\end{proof}
\end{sloppypar}

\(\llgBD{nask,get,tell}\) is not comparable with \(\llgBDD{ask,nask,tell}\).

\begin{myproposition} 
\label{dda-prop-17a}
\mbox{}

\noindent
\( \begin{array}[t]{lll}
\\
(i) & & \llgBDD{ask,nask,tell} \not\leq \llgBD{nask,get,tell}\\
(ii) & & \llgBD{nask,get,tell} \not\leq \llgBDD{ask,nask,tell}
\end{array} \)
\end{myproposition}

\begin{sloppypar}
\begin{proof}
\textbf{(i)} Otherwise, \(\llgBDD{ask,tell} \leq \llgBD{nask,get,tell}\) 
which contradicts proposition 12(i). 
\textbf{(ii)} Resulting from proposition 9(i).
\hfill 
\end{proof}
\end{sloppypar}


%
%
%
%

\begin{sloppypar}
\(\llgBDD{ask,nask,tell}\) is strictly less expressive than \(\llgBDD{ask,nask,get,tell}\).
\end{sloppypar}

\begin{myproposition} 
\label{dda-prop-18a}
\( \llgBDD{ask,nask,tell} < \llgBDD{ask,nask,get,tell}
\)
\end{myproposition}

\begin{sloppypar}
\begin{proof}
On the one hand,
\( \llgBDD{ask,nask,tell} \leq \llgBDD{ask,nask,get,tell}
\)
results from language inclusion. On the other hand,
\( \llgBDD{ask,nask,get,tell} \not\leq  \llgBDD{ask,nask,tell}.
\) Otherwise, \( \llgBDD{get,tell} \leq  \llgBDD{ask,nask,tell}\), which contradicts
proposition 13(vii).
\end{proof}
\end{sloppypar}

\(\llgBDD{get,tell}\) is strictly less expressive than \(\llgBDD{nask,get,tell}\).

\begin{myproposition} 
\label{dda-prop-19a}
\( \llgBDD{get,tell} < \llgBDD{nask,get,tell}
\)
\end{myproposition}

\begin{sloppypar}
\begin{proof}
On the one hand,
\( \llgBDD{get,tell} \leq \llgBDD{nask,get,tell}
\)
results from language inclusion. On the other hand,
\( \llgBDD{nask,get,tell} \not\leq  \llgBDD{get,tell}
\)
is established by contradiction.
Consider \( tell([t](1)) \seqc nask([t](1))\), 
for which 
\( \opsem(tell([t](1)) \seqc nask([t](1)) = \{ (\{ t(1) \},\delta^{-}) \}\). Hence, by
\textit{$P_{2}$} and \textit{$P_{3}$}, \(\coder(tell([t](1))) \seqc
\coder(nask([t](1)))\) fails. The contradiction comes then from the fact
that at least one computation of \(\coder(tell([t](1))) \seqc
\coder(nask([t](1)))\) starting on the empty store is
successful. Indeed, as \( \opsem(tell([t](1))) = \{ (\{ t(1) \},
\delta^{+}) \}\), any computation of $\coder(tell([t](1)))$
starting on the empty store succeeds. Similarly, any computation of
$\coder(nask([t](1)))$ starting on the empty store succeeds. 
Moreover, as $\coder(nask([t](1)))$
is composed of get and tell primitives only, for any store $\sigma$,
$\coder(nask([t](1)))$ admits at least one successful computation starting on $\sigma$. 
It follows that any
computation of $\coder(tell([t](1)))$ starting on the empty store can be
continued by a (successful) computation of $\coder(nask([t](1)))$, which
leads to the announced contradiction. \hfill 
\end{proof}
\end{sloppypar}

\begin{sloppypar}
Finally, \(\llgBD{ask,nask,get,tell}\) can be proved strictly less
expressive than \(\llgBDD{ask,nask,get,tell}\).
\end{sloppypar}

\begin{myproposition} 
\label{dda-prop-20a}
\( \llgBD{ask,nask,get,tell} < \llgBDD{ask,nask,get,tell}
\)
\end{myproposition}

\begin{sloppypar}
\begin{proof}
On the one hand,
\( \llgBD{ask,nask,get,tell} \leq \llgBDD{ask,nask,get,tell}
\) 
is directly deduced from proposition 2. On the other hand, 
if one had
\( \llgBDD{ask,nask,get,tell} \leq \llgBD{ask,nask,get,tell}
\)
then
\(\llgBDD{get,tell} \leq  \llgBD{nask,get,tell}\)
would hold, which contradicts proposition 13(v).
  \hfill 
\end{proof}
\end{sloppypar}


\section{Conclusion}
\label{sectConcl}

This paper is written in the continuity of our previous research on the expressiveness of Linda-like languages.
It has presented an extension of our Dense Bach language, that had promoted
the notions of density and dense tokens. The new language, called Dense Bach with
Distributed Density, proposes to distribute the density on a finite list
of tokens manipulated by the four classical primitives of the language. Technically
this is achieved by associating a positive number, called density, to a finite list of tokens
and to distribute this density among the tokens of the list.

Our work builds upon previous work by some of the authors
\cite{Br-Ja-COORD99,BrJa-SCP00,Jacquet-Linden-Foclasa03,%
  Jacquet-Linden-Coord04,Jacquet-Linden-Foclasa04}. We have
essentially followed the same lines and in particular have used de
Boer and Palamidessi's notion of modular embedding to compare the
families of sublanguages of Dense Bach and Dense Bach with Distributed
Density. Accordingly, we have established a gain of expressivity,
namely that Dense Bach with Distributed Density is strictly more
expressive than Dense Bach and, consequently, in view of the results
of \cite{Jacquet-Linden-Darquennes-Foclasa13}, strictly more
expressive than the Bach  and Linda languages.

Our work has similarities but also differences with several work on
the expressiveness of Linda-like languages. Compared to
\cite{Zavat01} and \cite{Zavat02}, it is worth observing that a
different comparison criteria is used to compare the expressiveness of
languages. Indeed, in these pieces of work, the comparison is
performed on (i) the compositionality of the encoding with respect to
parallel composition, (ii) the preservation of divergence and
deadlock, and (iii) a symmetry condition. Moreover, as will be
observed by the careful reader, we have taken a more liberal view with
respect to the preservation of termination marks in requiring these
preservations on the store resulting from the execution from the empty
store of the coded versions of the considered agents and not on the
same store. In particular, these ending stores are not required to be
of the form \(\sigma \cup \sigma\) (where \( \cup \) denotes multi-set
union) if this is so for the stores resulting from the agents
themselves.

In \cite{Zavat03}, nine variants of the \(\llg{ask,nask,get,tell} \)
language are studied. They are obtained by varying both the nature of
the shared data space and its structure. Rephrased in the setting of 
\cite{DB-PA-CONCUR-94}, this amounts to considering different
operational semantics. In contrast, in our work we fix an operational
semantics and compare different languages on the basis of this
semantics.  In \cite{Zavat04}, a process algebraic treatment of a
family of Linda-like concurrent languages is presented.  Again,
different semantics are considered whereas we have sticked to one
semantics and have compared languages on this basis.

In \cite{Zavat07}, a study of the absolute expressive power of
different variants of Linda-like languages has been made, whereas we
study the relative expressive power of different variants of such
languages (using modular embedding as a yard-stick and the ordered
interpretation of tell).  

It is worth observing that
\cite{Zavat01,Zavat02,Zavat03,Zavat04,Zavat07} do not deal with a
notion of density attached to tuples. In contrast, \cite{Zavat05} and
\cite{Zavat06} decorate tuples with an extra field in order to
investigate how probabilities and priorities can be introduced in the
Linda coordination model. Different expressiveness results are
established in \cite{Zavat05} but on an absolute level with respect to
Turing expressiveness and the possibility to encode the Leader
Election Problem. Our work contrasts in several aspects. First, we have
established relative expressiveness results by comparing the
sublanguages of two families. Moreover, some of these sublanguages
incorporate the $nask$ primitives, which
 strictly increases the
expressiveness. Finally, the introduction of density resembles but is
not identical to the association of weights to tuples. Indeed, in
contrast to \cite{Zavat05,Zavat06} we do not modify the tuples on
the store and do not modify the matching function so as to retrieve
the tuple with the highest weight. In contrast, we modify the tuple
primitives so as to be able to atomically put several occurrences of a
tuple on the store and check for the presence or absence of a number
of occurrences. We have also introduced a distribution of a density
among the tokens of a
set, which results in adding a new
non-deterministic behavior to the tell, ask and get primitives. As can
be appreciated by the reader through the comparison of Bach, Dense
Bach and Dense Bach with Distributed Density, this facility of
handling atomically several occurrences produces a real increase of
expressiveness.  One may however naturally think of encoding the
number of occurrences of a tuple as an additional weight-like
parameter. It is nevertheless not clear how our primitives tackling at
once several occurrences can be rephrased in Linda-like primitives and
how the induced encoding would still fulfills the requirements of
modularity. 
Moreover, in contrast to Linda-like language, due to the
non-determinism of the get and tell primitives, it is not clear how to
code ask primitives by get and tell ones in our distributed density
framework.
 This will be the subject for future research.

In \cite{Viroli09}, Viroli and Casadei propose a stochastic extension
of the Linda framework, with a notion of tuple concentration, similar
to the weight of \cite{Zavat05} and \cite{Zavat06} and our notion of
density. The syntax of this tuple space is modeled by means of a
calculus, with an operational semantics given as an hybrid CTMC/DTMC
model. This operational semantics describes the behavior of tell, ask
and get like primitives but does not consider a nask like
primitive. Moreover, no expressiveness results are established and
there is no counterpart for non-determinism arising from the
distribution of density on tokens.

These three last pieces of work tackle probabilistic extensions of
Linda-like languages.  As a further and natural step in our research,
we aim at studying how our notion of density can be the basis of such
probabilistic extensions. As our work also relies on the possibility
to atomically put several occurrences of tokens and test for their
presence or absence, we will also examine in future work how Dense
Bach with Distributed Density compares with the Gamma language.

%

\bibliographystyle{eptcs}
\bibliography{dda}

\begin{thebibliography}{10}
\providecommand{\bibitemdeclare}[2]{}
\providecommand{\surnamestart}{}
\providecommand{\surnameend}{}
\providecommand{\urlprefix}{Available at }
\providecommand{\url}[1]{\texttt{#1}}
\providecommand{\href}[2]{\texttt{#2}}
\providecommand{\urlalt}[2]{\href{#1}{#2}}
\providecommand{\doi}[1]{doi:\urlalt{http://dx.doi.org/#1}{#1}}
\providecommand{\bibinfo}[2]{#2}

\bibitemdeclare{article}{DB-PA-CONCUR-94}
\bibitem{DB-PA-CONCUR-94}
\bibinfo{author}{F.S. \surnamestart de~Boer\surnameend} \&
  \bibinfo{author}{C.~\surnamestart Palamidessi\surnameend}
  (\bibinfo{year}{1994}): \emph{\bibinfo{title}{Embedding as a {T}ool for
  {L}anguage {C}omparison}}.
\newblock {\sl \bibinfo{journal}{Information and Computation}}
  \bibinfo{volume}{108}(\bibinfo{number}{1}), pp. \bibinfo{pages}{128--157},
  \doi{10.1006/inco.1994.1004}.

\bibitemdeclare{inproceedings}{Zavat03}
\bibitem{Zavat03}
\bibinfo{author}{Marcello~M. \surnamestart Bonsangue\surnameend},
  \bibinfo{author}{Joost~N. \surnamestart Kok\surnameend} \&
  \bibinfo{author}{Gianluigi \surnamestart Zavattaro\surnameend}
  (\bibinfo{year}{1999}): \emph{\bibinfo{title}{{Comparing coordination models
  based on shared distributed replicated data}}}.
\newblock In: {\sl \bibinfo{booktitle}{ACM Symposium on Applied Computing}},
  pp. \bibinfo{pages}{156--165}, \doi{10.1145/298151.298226}.

\bibitemdeclare{article}{Zavat05}
\bibitem{Zavat05}
\bibinfo{author}{M.~\surnamestart Bravetti\surnameend},
  \bibinfo{author}{R.~\surnamestart Gorrieri\surnameend},
  \bibinfo{author}{R.~\surnamestart Lucchi\surnameend} \&
  \bibinfo{author}{\surnamestart G.Zavattaro\surnameend}
  (\bibinfo{year}{2005}): \emph{\bibinfo{title}{{Quantitative Information in
  the Tuple Space Coordination Model}}}.
\newblock {\sl \bibinfo{journal}{Theoretical Computer Science}}
  \bibinfo{volume}{346}(\bibinfo{number}{1}), pp. \bibinfo{pages}{28--57},
  \doi{10.1016/j.tcs.2005.08.004}.

\bibitemdeclare{inproceedings}{Zavat06}
\bibitem{Zavat06}
\bibinfo{author}{M.~\surnamestart Bravetti\surnameend},
  \bibinfo{author}{R.~\surnamestart Gorrieri\surnameend},
  \bibinfo{author}{R.~\surnamestart Lucchi\surnameend} \&
  \bibinfo{author}{G.~\surnamestart Zavattaro\surnameend}
  (\bibinfo{year}{2004}): \emph{\bibinfo{title}{{Probabilistic and Prioritized
  Data Retrieval in the Linda Coordination Model}}}.
\newblock In \bibinfo{editor}{R.~\surnamestart {De Nicola}\surnameend},
  \bibinfo{editor}{G.L. \surnamestart Ferrari\surnameend} \&
  \bibinfo{editor}{G.~\surnamestart Meredith\surnameend}, editors: {\sl
  \bibinfo{booktitle}{Proceedings of the 6th International Conference on
  Coordination Models and Languages}}, {\sl \bibinfo{series}{Lecture Notes in
  Computer Science}} \bibinfo{volume}{2949}, \bibinfo{publisher}{Springer}, pp.
  \bibinfo{pages}{55--70}, \doi{10.1007/978-3-540-24634-3_7}.

\bibitemdeclare{inproceedings}{Br-Ja-COORD97}
\bibitem{Br-Ja-COORD97}
\bibinfo{author}{A.~\surnamestart Brogi\surnameend} \& \bibinfo{author}{J.-M.
  \surnamestart Jacquet\surnameend} (\bibinfo{year}{1997}):
  \emph{\bibinfo{title}{Modeling {C}oordination via {A}synchronous
  {C}ommunication}}.
\newblock In \bibinfo{editor}{D.~\surnamestart Garlan\surnameend} \&
  \bibinfo{editor}{D.~\surnamestart {Le M\'etayer}\surnameend}, editors: {\sl
  \bibinfo{booktitle}{Proceedings of the {S}econd {I}nternational {C}onference
  on {C}oordination {L}anguages and {M}odels}}, {\sl \bibinfo{series}{Lecture
  Notes in Computer Science}} \bibinfo{volume}{1282}, \bibinfo{address}{Berlin,
  Germany}, pp. \bibinfo{pages}{238--255}, \doi{10.1007/3-540-63383-9_84}.

\bibitemdeclare{article}{Br-Ja-Express98-ii}
\bibitem{Br-Ja-Express98-ii}
\bibinfo{author}{A.~\surnamestart Brogi\surnameend} \& \bibinfo{author}{J.-M.
  \surnamestart Jacquet\surnameend} (\bibinfo{year}{1998}):
  \emph{\bibinfo{title}{On the {E}xpressiveness of {L}inda-like {C}oncurrent
  {L}anguages}}.
\newblock {\sl \bibinfo{journal}{Electronic Notes in Theoretical Computer
  Science}} \bibinfo{volume}{16}(\bibinfo{number}{2}), pp.
  \bibinfo{pages}{61--82}, \doi{10.1016/S1571-0661(04)00118-5}.

\bibitemdeclare{inproceedings}{Br-Ja-COORD99}
\bibitem{Br-Ja-COORD99}
\bibinfo{author}{A.~\surnamestart Brogi\surnameend} \& \bibinfo{author}{J.-M.
  \surnamestart Jacquet\surnameend} (\bibinfo{year}{1999}):
  \emph{\bibinfo{title}{On the {E}xpressiveness of {C}oordination {M}odels}}.
\newblock In \bibinfo{editor}{C.~\surnamestart Ciancarini\surnameend} \&
  \bibinfo{editor}{A.~\surnamestart Wolf\surnameend}, editors: {\sl
  \bibinfo{booktitle}{Proceedings of the Third International Conference on
  Coordination Languages and Models}}, {\sl \bibinfo{series}{Lecture Notes in
  Computer Science}} \bibinfo{volume}{1594},
  \bibinfo{publisher}{Springer-Verlag}, pp. \bibinfo{pages}{134--149},
  \doi{10.1007/3-540-48919-3_11}.

\bibitemdeclare{article}{BrJa-SCP00}
\bibitem{BrJa-SCP00}
\bibinfo{author}{A.~\surnamestart Brogi\surnameend} \& \bibinfo{author}{J.-M.
  \surnamestart Jacquet\surnameend} (\bibinfo{year}{2003}):
  \emph{\bibinfo{title}{{On the Expressiveness of Coordination via Shared
  Dataspaces}}}.
\newblock {\sl \bibinfo{journal}{Science of Computer Programming}}
  \bibinfo{volume}{46}(\bibinfo{number}{1--2}), pp. \bibinfo{pages}{71 -- 98},
  \doi{10.1016/S0167-6423(02)00087-4}.

\bibitemdeclare{article}{BrJaLi-Foclasa02}
\bibitem{BrJaLi-Foclasa02}
\bibinfo{author}{A.~\surnamestart Brogi\surnameend}, \bibinfo{author}{J.-M.
  \surnamestart Jacquet\surnameend} \& \bibinfo{author}{I.~\surnamestart
  Linden\surnameend} (\bibinfo{year}{2003}): \emph{\bibinfo{title}{{On Modeling
  Coordination via Asynchronous Communication and Enhanced Matching}}}.
\newblock {\sl \bibinfo{journal}{Electronic Notes in Theoretical Computer
  Science}} \bibinfo{volume}{68}(\bibinfo{number}{3}),
  \doi{10.1016/S1571-0661(05)82568-X}.

\bibitemdeclare{article}{Zavat07}
\bibitem{Zavat07}
\bibinfo{author}{Nadia \surnamestart Busi\surnameend}, \bibinfo{author}{Roberto
  \surnamestart Gorrieri\surnameend} \& \bibinfo{author}{Gianluigi
  \surnamestart Zavattaro\surnameend} (\bibinfo{year}{1997}):
  \emph{\bibinfo{title}{{On the Turing equivalence of Linda coordination
  primitives}}}.
\newblock {\sl \bibinfo{journal}{Electronic Notes in Theoretical Computer
  Science}} \bibinfo{volume}{7}, pp. \bibinfo{pages}{75--75},
  \doi{10.1016/S1571-0661(05)80467-0}.

\bibitemdeclare{article}{Zavat04}
\bibitem{Zavat04}
\bibinfo{author}{Nadia \surnamestart Busi\surnameend}, \bibinfo{author}{Roberto
  \surnamestart Gorrieri\surnameend} \& \bibinfo{author}{Gianluigi
  \surnamestart Zavattaro\surnameend} (\bibinfo{year}{1998}):
  \emph{\bibinfo{title}{{A Process Algebraic View of Linda Coordination
  Primitives}}}.
\newblock {\sl \bibinfo{journal}{Theoretical Computer Science}}
  \bibinfo{volume}{192}, pp. \bibinfo{pages}{167--199},
  \doi{10.1016/S0304-3975(97)00149-7}.

\bibitemdeclare{inproceedings}{Jacquet-Linden-Darquennes-Foclasa13}
\bibitem{Jacquet-Linden-Darquennes-Foclasa13}
\bibinfo{author}{J.-M. \surnamestart Jacquet\surnameend},
  \bibinfo{author}{I.~\surnamestart Linden\surnameend} \&
  \bibinfo{author}{D.~\surnamestart Darquennes\surnameend}
  (\bibinfo{year}{2013}): \emph{\bibinfo{title}{On {D}ensity in {C}oordination
  {L}anguages}}.
\newblock In \bibinfo{editor}{C.~\surnamestart Canal\surnameend} \&
  \bibinfo{editor}{M.~\surnamestart Villari\surnameend}, editors: {\sl
  \bibinfo{booktitle}{Proceedings of the European Conference on Service
  Oriented and Cloud Computing 2013}}, {\sl \bibinfo{series}{Communications in
  Computer and Information Science}} \bibinfo{volume}{393},
  \bibinfo{publisher}{Springer}, pp. \bibinfo{pages}{189--203},
  \doi{10.1007/978-3-642-45364-9_16}.

\bibitemdeclare{inproceedings}{Jacquet-Linden-Coord04}
\bibitem{Jacquet-Linden-Coord04}
\bibinfo{author}{I.~\surnamestart Linden\surnameend} \& \bibinfo{author}{J.-M.
  \surnamestart Jacquet\surnameend} (\bibinfo{year}{2004}):
  \emph{\bibinfo{title}{On the {E}xpressiveness of {A}bsolute-{T}ime
  {C}oordination {L}anguages}}.
\newblock In \bibinfo{editor}{R.~De \surnamestart Nicola\surnameend},
  \bibinfo{editor}{G.L. \surnamestart Ferrari\surnameend} \&
  \bibinfo{editor}{G.~\surnamestart Meredith\surnameend}, editors: {\sl
  \bibinfo{booktitle}{Proc. 6th International Conference on Coordination Models
  and Languages}}, {\sl \bibinfo{series}{Lecture Notes in Computer Science}}
  \bibinfo{volume}{2949}, \bibinfo{publisher}{Springer}, pp.
  \bibinfo{pages}{232--247}, \doi{10.1007/978-3-540-24634-3_18}.

\bibitemdeclare{article}{Jacquet-Linden-Foclasa04}
\bibitem{Jacquet-Linden-Foclasa04}
\bibinfo{author}{I.~\surnamestart Linden\surnameend} \& \bibinfo{author}{J.-M.
  \surnamestart Jacquet\surnameend} (\bibinfo{year}{2007}):
  \emph{\bibinfo{title}{On the {E}xpressiveness of {T}imed {C}oordination via
  {S}hared {D}ataspaces}}.
\newblock {\sl \bibinfo{journal}{Electronical Notes in Theoretical Computer
  Science}} \bibinfo{volume}{180}(\bibinfo{number}{2}), pp.
  \bibinfo{pages}{71--89}, \doi{10.1016/j.entcs.2006.10.047}.

\bibitemdeclare{article}{Jacquet-Linden-Foclasa03}
\bibitem{Jacquet-Linden-Foclasa03}
\bibinfo{author}{I.~\surnamestart Linden\surnameend}, \bibinfo{author}{J.-M.
  \surnamestart Jacquet\surnameend}, \bibinfo{author}{K.~De \surnamestart
  Bosschere\surnameend} \& \bibinfo{author}{A.~\surnamestart Brogi\surnameend}
  (\bibinfo{year}{2004}): \emph{\bibinfo{title}{On the {E}xpressiveness of
  {R}elative-{T}imed {C}oordination {M}odels}}.
\newblock {\sl \bibinfo{journal}{Electronical Notes in Theoretical Computer
  Science}} \bibinfo{volume}{97}, pp. \bibinfo{pages}{125--153},
  \doi{10.1016/j.entcs.2004.04.034}.

\bibitemdeclare{article}{LindenJBB06}
\bibitem{LindenJBB06}
\bibinfo{author}{I.~\surnamestart Linden\surnameend}, \bibinfo{author}{J.-M.
  \surnamestart Jacquet\surnameend}, \bibinfo{author}{K.~De \surnamestart
  Bosschere\surnameend} \& \bibinfo{author}{A.~\surnamestart Brogi\surnameend}
  (\bibinfo{year}{2006}): \emph{\bibinfo{title}{On the {E}xpressiveness of
  {T}imed {C}oordination {M}odels}}.
\newblock {\sl \bibinfo{journal}{Science of Computer Programming}}
  \bibinfo{volume}{61}(\bibinfo{number}{2}), pp. \bibinfo{pages}{152--187},
  \doi{10.1016/j.scico.2005.10.011}.

\bibitemdeclare{inproceedings}{Shapiro}
\bibitem{Shapiro}
\bibinfo{author}{E.Y. \surnamestart Shapiro\surnameend} (\bibinfo{year}{1992}):
  \emph{\bibinfo{title}{Embeddings {A}mong {C}oncurrent {P}rogramming
  {L}anguages}}.
\newblock In \bibinfo{editor}{W.R. \surnamestart Cleaveland\surnameend},
  editor: {\sl \bibinfo{booktitle}{Proceedings of Concur 1992}},
  \bibinfo{series}{Lecture Notes in Computer Science},
  \bibinfo{publisher}{Springer}, pp. \bibinfo{pages}{486--503},
  \doi{10.1007/BFb0084811}.

\bibitemdeclare{inproceedings}{Viroli09}
\bibitem{Viroli09}
\bibinfo{author}{M.~\surnamestart Viroli\surnameend} \&
  \bibinfo{author}{M.~\surnamestart Casadei\surnameend} (\bibinfo{year}{2009}):
  \emph{\bibinfo{title}{Biochemical {T}uple {S}paces for {S}elf-organising
  {C}oordination}}.
\newblock In \bibinfo{editor}{J.~\surnamestart Field\surnameend} \&
  \bibinfo{editor}{V.~T. \surnamestart Vasconcelos\surnameend}, editors: {\sl
  \bibinfo{booktitle}{Proceedings of 11th International Conference on
  Coordination Models and Languages}}, {\sl \bibinfo{series}{Lecture Notes in
  Computer Science}} \bibinfo{volume}{5521}, \bibinfo{publisher}{Springer}, pp.
  \bibinfo{pages}{143--162}, \doi{10.1007/978-3-642-02053-7_8}.

\bibitemdeclare{article}{Zavat01}
\bibitem{Zavat01}
\bibinfo{author}{G.~\surnamestart Zavattaro\surnameend} (\bibinfo{year}{1998}):
  \emph{\bibinfo{title}{{On the incomparability of Gamma and Linda}}}.
\newblock {\sl \bibinfo{journal}{Electronic Transactions on Numerical
  Analysis}}.

\bibitemdeclare{article}{Zavat02}
\bibitem{Zavat02}
\bibinfo{author}{Gianluigi \surnamestart Zavattaro\surnameend}
  (\bibinfo{year}{1998}): \emph{\bibinfo{title}{{Towards a Hierarchy of
  Negative Test Operators for Generative Communication}}}.
\newblock {\sl \bibinfo{journal}{Electronic Notes in Theoretical Computer
  Science}} \bibinfo{volume}{16}, pp. \bibinfo{pages}{154--170},
  \doi{10.1016/S1571-0661(04)00125-2}.

\end{thebibliography}

\end{document}